\documentclass[a4paper,onecolumn,10pt,accepted=2021-07-29]{quantumarticle}
\pdfoutput=1

\usepackage[dvipsnames]{xcolor}
\usepackage[colorlinks=true,bookmarks=false]{hyperref}
\usepackage{amsmath}
\usepackage{comment}
\usepackage{amsfonts}
\usepackage{tabularx}
\usepackage{graphicx}
\usepackage{mathtools}
\usepackage{braket}
\usepackage{amsthm}
\usepackage{mathtools}
\usepackage[sort&compress,numbers]{natbib}
\numberwithin{equation}{section}

\newtheorem{theorem}{Theorem}
\newtheorem{example}{Example}

\newtheorem{proposition}{Proposition}

\theoremstyle{remark}

\newcommand{\real}{\operatorname{Re}}

\newcommand{\parti}[2]{\frac{\partial #1}{\partial #2}}

\newcommand{\intall}{\int_{-\infty}^{\infty}}

\newcommand{\Avg}[1]{\left\langle#1\right\rangle}

\newcommand{\bk}[1]{\left(#1\right)}
\newcommand{\Bk}[1]{\left[#1\right]}
\newcommand{\BK}[1]{\left\{#1\right\}}
\newcommand{\trace}{\operatorname{tr}}

\newcommand{\expect}{\mathbb E}

\newcommand{\supp}{\operatorname{supp}}
\begin{document}

\title{Poisson Quantum Information}

\author{Mankei Tsang}
\affiliation{Department of Electrical and Computer Engineering,
  National University of Singapore, 4 Engineering Drive 3, Singapore
  117583}

\affiliation{Department of Physics, National University of Singapore,
  2 Science Drive 3, Singapore 117551}

\email{mankei@nus.edu.sg}
\homepage{https://blog.nus.edu.sg/mankei/}
\orcid{0000-0001-7173-1239}

%\date{\today}

\maketitle

\begin{abstract}
  By taking a Poisson limit for a sequence of rare quantum objects,
  I derive simple formulas for the Uhlmann fidelity, the quantum
  Chernoff quantity, the relative entropy, and the Helstrom
  information. I also present analogous formulas in classical
  information theory for a Poisson model. An operator called the
  intensity operator emerges as the central quantity in the formalism
  to describe Poisson states.  It behaves like a density operator but
  is unnormalized. The formulas in terms of the intensity operators
  not only resemble the general formulas in terms of the density
  operators, but also coincide with some existing definitions of
  divergences between unnormalized positive-semidefinite matrices.
  Furthermore, I show that the effects of certain channels on Poisson
  states can be described by simple maps for the intensity operators.
\end{abstract}

\section{Introduction}

The Poisson limit theorems, also called the laws of rare events or the
laws of small numbers, underlie the ubiquity of Poisson statistics and
enable a variety of simplifications in probability theory
\cite{falk11,snyder_miller}. In this paper, I examine the consequences
of taking a similar limit in quantum information theory and
demonstrate that elegant formulas emerge under the limit. This body of
work may be called Poisson quantum information, which has the
potential to grow into a fruitful research topic on par with Gaussian
quantum information \cite{weedbrook,holevo19}.

The results put forth are a generalization and culmination of our
earlier efforts concerning weak thermal light
\cite{stellar,tnl,tnl2,tsang19a}. The results are especially relevant
to the recent literature on the application of quantum information
theory to partially coherent imaging
\cite{tsang19a,larson18,tsang_comment19,larson19,hradil19,liang21,wadood21,hradil21,de21,
  kurdzialek21}, where there is some confusion regarding the correct
formulas for the information quantities, leading to inconsistent
results by different groups. This paper resolves the debate in support
of Refs.~\cite{tnl,tsang_comment19} and extends the results, beyond
the Helstrom and Fisher information quantities considered there. While
special weak-thermal-light models have already found success by
enabling substantial simplifications in many previous studies in
quantum optics
\cite{helstrom,stellar,tnl,tsang19a,lu18,gottesman,khabiboulline19},
the general theory here stands on its own and does not require the
quantum state to be bosonic or exactly thermal, so it may be applied
to particles and systems beyond photons, such as electrons and
quasiparticles, whenever a quantum treatment of rare objects is
needed.

\section{\label{sec_poisson}Poisson states}
Let $\mathcal P(\mathcal H)$ be the set of positive-semidefinite
operators on a Hilbert space $\mathcal H$ and
$\mathcal P_1(\mathcal H)$ be the unit-trace subset
$\mathcal P_1(\mathcal H) \equiv \{\tau \in \mathcal P(\mathcal
H)|\trace \tau = 1\}$.  Consider the Hilbert space
$(\mathcal H_0\oplus\mathcal H_1)^{\otimes M}$ for $M$ temporal modes,
where $\mathcal H_0$ is a 1-dimensional vacuum Hilbert space and
$\mathcal H_1$ is a $d$-dimensional Hilbert space for a quantum
object.  Let the density operator
$\rho_M \in \mathcal P_1[(\mathcal H_0\oplus\mathcal H_1)^{\otimes
  M}]$ be
\begin{align}
\rho_M &= \tau^{\otimes M},
&
\tau &= \bigoplus_{l=0}^1 \pi_l\tau_l,
&
\pi_0 &= 1-\epsilon,
&
\pi_1 &= \epsilon,
\label{rare}
\end{align}
where $\tau_0 \in \mathcal P_1(\mathcal H_0)$ is the vacuum state,
$\tau_1 \in \mathcal P_1(\mathcal H_1)$ is the density operator for
the quantum object, and $0\le \epsilon \le 1$ is the probability of
having one object in each temporal mode. See, for example,
Ref.~\cite{watrous} for the definitions of the mathematical
concepts. In the context of optics, $\tau_1$ may denote the one-photon
density operator in $d$ spatial and polarization modes.  The
formalism, in itself, is generic however---the object can be any
elementary or composite quantum system. For example, to study
intensity interferometry \cite{mandel}, where two-photon-coincidence
events are postselected and one-photon events are ignored, $\tau_1$
can be used to denote the two-photon density operator.

Define the Poisson limit as
\begin{align}
\epsilon &\to 0, 
& M &\to \infty, &  
N &\equiv M\epsilon\quad\textrm{fixed},
\label{poisson_limit}
\end{align}
where $N$ is the expected object number in total.  In this limit, a
central quantity in the ensuing theory is
\begin{align}
\Gamma &\equiv N \tau_1 \in \mathcal P(\mathcal H_1),
\label{intensity_op}
\end{align}
which I call the intensity operator. Similar to $\tau_1$, it is a
positive-semidefinite operator on $\mathcal H_1$. Unlike $\tau_1$,
however, its trace
\begin{align}
\trace\Gamma &= N
\end{align}
is not normalized, as the object number need not be conserved in a
problem.

Let
$\mathcal N^{(k)} \equiv (\mathcal N^{(k)}_1,\mathcal
N^{(k)}_2,\dots)$ be the vectoral random variable from a multi-output
object-counting measurement of the $k$th temporal mode.  Let its
probability distribution be
\begin{align}
P_{\mathcal N^{(k)}}(0,\dots,0) &= \trace I_0 \tau
= 1 - \epsilon,
\label{P0}
\\
P_{\mathcal N^{(k)}}(0,\dots,n_j = 1,\dots,0)  &= 
\trace E_j\tau
= \epsilon \trace E_j\tau_1,
\label{P1}
\end{align}
where $I_l$ is the projection operator into $\mathcal H_l$ and $E$ is
a positive operator-valued measure (POVM) on $\mathcal H_1$ that
satisfies $E_j \in \mathcal P(\mathcal H_1)$ and $\sum_{j} E_j =
I_1$. Under the Poisson limit, the integrated random variable
$\mathcal M \equiv \sum_{k=1}^M \mathcal N^{(k)}$ has the Poisson
distribution \cite{falk11,vaart96}
\begin{align}
P_{\mathcal M}(m) &= \prod_{j} \exp(-\Lambda_j) \frac{\Lambda_j^{m_j}}{m_j!},
\label{poisson}
\end{align}
where each intensity value is given by
\begin{align}
\Lambda_j &= \trace E_j\Gamma.
\label{intensity}
\end{align}
Appendix~\ref{app_poisson} states, in more rigorous terms, a Poisson
limit theorem that gives Eqs.~(\ref{poisson}) and (\ref{intensity}).

In view of the Poissonian properties, the quantum state given by
Eqs.~(\ref{rare}), together with the Poisson limit given by
Eqs.~(\ref{poisson_limit}), may be called a Poisson state, denoted as
$\rho$ without the subscript $M$. Although a rigorous treatment of
Poisson states may require quantum stochastic calculus
\cite{parthasarathy92} or nonstandard analysis
\cite{leitz01,nelson87}, in the following I take the more convenient
approach of assuming Eqs.~(\ref{rare}) and taking the Poisson limit at
the end of a calculation.  Appendix~\ref{app_state} discusses how
Poisson states may be defined more rigorously in terms of Fock spaces.

\section{\label{sec_info}Information quantities}
To define information quantities, assume two Poisson states $\rho$ and
$\rho'$.  Assume that $M$ and $\tau_0$ are the same for the two
states, while the other quantities may vary. The quantities for the
second state are denoted with a prime; for example, the intensity
operator is $\Gamma'$ and $N' \equiv M\epsilon' =
\trace\Gamma'$. Assume further that $\epsilon' = O(\epsilon)$ (order
at most $\epsilon$ in the Poisson limit), such that the Poisson limit
applies to both states. My first proposition concerns the Uhlmann
fidelity \cite{uhlmann_crell}, the most notable application of which
is to set useful bounds for quantum hypothesis testing \cite{hayashi}.
\begin{proposition}
\label{prop_uhlmann}
The Uhlmann fidelity 
\begin{align}
F(\rho,\rho') &\equiv \trace\sqrt{\sqrt{\rho}\rho'\sqrt{\rho}}
\label{uhlmann}
\end{align}
between two Poisson states is given by
\begin{align}
F(\rho,\rho') &= \exp\Bk{-\frac{N+N'}{2} + F(\Gamma,\Gamma')}.
\label{uhlmann_poisson}
\end{align}
\end{proposition}
\begin{proof}
  Given Eqs.~(\ref{rare}) and using basic linear algebra (see, for
  example, Refs.~\cite[Chap.~1]{watrous} and
  \cite[Proposition~19.1]{parthasarathy92}), it can be shown that
\begin{align}
F(\rho_M,\rho_M') &= [F(\tau,\tau')]^M
= \Bk{\sum_l F(\pi_l\tau_l,\pi_l'\tau_l')}^M
= \Bk{1 - \frac{\epsilon+\epsilon'}{2} + \sqrt{\epsilon\epsilon'}
F(\tau_1,\tau_1') + O(\epsilon^2)}^M.
\end{align}
Taking the Poisson limit then leads to the proposition.
\end{proof}
Equation~(\ref{uhlmann_poisson}) has a self-similar feature: it
contains a fidelity expression, in terms of the intensity operators
$\Gamma$ and $\Gamma'$, that has the same form as the general formula
for the density operators. Another remarkable feature is that the
$d_B^2$ quantity, defined in the following expression
\begin{align}
-2\ln F(\rho,\rho')
&= N +N' - 2F(\Gamma,\Gamma') \equiv d_B^2(\Gamma,\Gamma'),
\label{bures}
\end{align}
coincides with the squared Bures-Wasserstein distance between
unnormalized positive-semidefinite operators
\cite{uhlmann_crell,bhatia19}.

The next two propositions concern the quantum Chernoff quantity and
the relative entropy. The Chernoff quantity is used in the quantum
Chernoff bound for hypothesis testing
\cite{audenaert,nussbaum09,hayashi}, while the relative entropy is, of
course, a fundamental quantity in quantum thermodynamics
\cite{wehrl78} and communication theory \cite{hayashi,holevo19}.
% The proofs are similar to that of Proposition~\ref{prop_uhlmann} and
% omitted for brevity.
\begin{proposition}
\label{prop_qchernoff}
The quantum Chernoff quantity
\begin{align}
C_s(\rho,\rho') &\equiv \trace \rho^{s} \rho'^{1-s},
\quad
0 \le s \le 1
\label{qchernoff}
\end{align}
for two Poisson states is given by
\begin{align}
C_s(\rho,\rho') &= \exp\Bk{-sN-(1-s)N' + C_s(\Gamma,\Gamma')}.
\label{qchernoff_poisson}
\end{align}
The quantum Chernoff distance \cite{audenaert,nussbaum09} is then
\begin{align}
-\ln\Bk{\inf_{0\le s \le 1} C_s(\rho,\rho')} &= 
\sup_{0\le s \le 1} \Bk{sN+(1-s)N' - C_s(\Gamma,\Gamma')}.
\label{qchernoff_distance}
\end{align}
\end{proposition}

\begin{proposition}
\label{prop_KL}
The relative entropy
\begin{align}
D(\rho\|\rho') &\equiv 
\trace \rho'-\trace\rho + \trace \rho \bk{\ln \rho-\ln\rho'}
\label{KL}
\end{align}
between two Poisson states is given by
\begin{align}
D(\rho\|\rho') &= D(\Gamma\|\Gamma').
\label{KLpoisson}
\end{align}
If $\supp\Gamma \not\subseteq \supp\Gamma'$, where $\supp$ denotes the
support \cite{holevo19}, then $\supp\rho \not\subseteq \supp \rho'$,
and $D(\rho\|\rho') = D(\Gamma\|\Gamma') = \infty$.
\end{proposition}
The proofs of Propositions~\ref{prop_qchernoff} and \ref{prop_KL} are
delegated to Appendix~\ref{app_proofs}.

In Eq.~(\ref{KL}), the relative entropy is expressed in a more general
form so that it is appropriate for unnormalized positive-definite
operators as well \cite{lindblad73,dhillon07,amari16}. Similar to
Eq.~(\ref{uhlmann_poisson}), Eqs.~(\ref{qchernoff_poisson}) and
(\ref{KLpoisson}) have a self-similar feature. It is also remarkable
that the $D_s(\Gamma,\Gamma')$ and $D(\Gamma\|\Gamma')$ quantities,
defined in the following expressions
\begin{align}
-\frac{1}{s(1-s)}\ln C_s(\rho,\rho')
&= \frac{1}{s(1-s)}\Bk{sN + (1-s)N' - C_s(\Gamma,\Gamma')} \equiv 
D_s(\Gamma,\Gamma'),
\label{alpha_div_Gamma}
\\
D(\Gamma\|\Gamma') &\equiv N'-N + \trace \Gamma\bk{\ln\Gamma-\ln\Gamma'},
\label{KL_Gamma}
\end{align}
coincide with the alpha-divergences between unnormalized
positive-semidefinite matrices in the matrix-analysis literature
\cite{dhillon07,amari16}---$D_s(\Gamma,\Gamma')$ here is identical to
Eq.~(4.179) in Ref.~\cite{amari16} if one sets $s = (1-\alpha)/2$ and
$D(\Gamma\|\Gamma')$ here is identical to Eq.~(4.166) in
Ref.~\cite{amari16}.

The classical Poisson model given by Eq.~(\ref{poisson}) leads to
analogous formulas for the corresponding quantities in classical
information theory. The proofs are trivial and omitted for brevity;
see, for example, Ref.~\cite{kailath67} for similar results.  Assume
two Poisson distributions and again denote the quantities for the
second distribution with a prime.
\begin{proposition}
\label{prop_chernoff}
The Chernoff quantity
\begin{align}
C_s(P_{\mathcal M},P_{\mathcal M}') &\equiv 
 \sum_m P_{\mathcal M}^s P_{\mathcal M}'^{1-s},
\quad
0 \le s \le 1
\end{align}
for two Poisson distributions is given by
\begin{align}
C_s(P_{\mathcal M},P_{\mathcal M}') &= 
\exp\Bk{-sN-(1-s)N' + C_s(\Lambda,\Lambda')}.
\label{Bpoi}
\end{align}
The Chernoff distance  is then
\begin{align}
-\ln\Bk{\inf_{0\le s \le 1} C_s(P_{\mathcal M},P_{\mathcal M}')} &= 
\sup_{0\le s \le 1} \Bk{sN+(1-s)N' - C_s(\Lambda,\Lambda')}.
\label{chernoff_distance}
\end{align}
\end{proposition}
\begin{proposition}
\label{prop_KLclassical}
The relative entropy
\begin{align}
D(P_{\mathcal M}\|P_{\mathcal M}') &\equiv  \sum_m \bk{P_{\mathcal M}'-P_{\mathcal M}
+ P_{\mathcal M} \ln \frac{P_{\mathcal M}}{P_{\mathcal M}'}}
\end{align}
between two Poisson distributions is given by
\begin{align}
D(P_{\mathcal M}\|P_{\mathcal M}') &= D(\Lambda\|\Lambda').
\label{KLpoi}
\end{align}
If $\supp \Lambda \not\subseteq \supp \Lambda'$, then
$\supp P_{\mathcal M} \not\subseteq \supp P_{\mathcal M}'$, and
$D(P_{\mathcal M}\|P_{\mathcal M}') = D(\Lambda\|\Lambda') = \infty$.
\end{proposition}
Equations~(\ref{Bpoi}) and (\ref{KLpoi}) resemble
Eqs.~(\ref{qchernoff_poisson}) and (\ref{KLpoisson}) and also possess
a self-similar feature. Similar to the quantum case, the
$D_s(\Lambda,\Lambda')$ and $D(\Lambda\|\Lambda')$ quantities, defined
in the following expressions
\begin{align}
-\frac{1}{s(1-s)}\ln C_s(P_{\mathcal M},P_{\mathcal M}')
&= \frac{1}{s(1-s)}\Bk{s N + (1-s)N' - C_s(\Lambda,\Lambda')}
\equiv 
D_s(\Lambda,\Lambda'),
\label{alpha_div}
\\
D(\Lambda\|\Lambda') &\equiv 
N' - N + \sum_j \Lambda_j \ln \frac{\Lambda_j}{\Lambda_j'},
\end{align}
coincide with the alpha-divergences between unnormalized positive
distributions \cite{cichocki10}.  $D_{1/2}(\Lambda,\Lambda')/2$, in
particular, is the squared Hellinger distance \cite{bhatia19}.

A fundamental relation between the quantum and classical information
quantities is monotonicity. It is interesting to note that it
manifests for the Poisson states on two levels: on the level of $\rho$
and on the level of $\Gamma$.
\begin{proposition}
\label{prop_mono}
\begin{align}
F(\rho,\rho') &\le C_{1/2}(P_{\mathcal M},P_{\mathcal M}'),
\label{fidelity_bound}
\\
C_s(\rho,\rho') &\le C_s(P_{\mathcal M},P_{\mathcal M}'),
\\
D(\rho\|\rho') &\ge D(P_{\mathcal M}\|P_{\mathcal M}').
\label{KL_bound}
\end{align}
\end{proposition}
\begin{proof}
  These bounds follow directly from the monotonicity relations
  \cite{hayashi}. Alternatively, one may take $\Gamma = N \tau_1$ and
  $\Gamma' = N'\tau_1'$, define
  $p_j \equiv \trace E_j\tau_1 = \Lambda_j/N$ and
  $p_j' \equiv \trace E_j\tau_1' = \Lambda_j'/N'$, and apply
  the monotonicity relations with respect to $(\tau_1,\tau_1')$ and
  $(p,p')$ in order to obtain
\begin{align}
F(\Gamma,\Gamma')&= \sqrt{NN'}F(\tau_1,\tau_1') 
\le \sqrt{NN'}C_{1/2}(p,p') = C_{1/2}(\Lambda,\Lambda'),
\\
C_s(\Gamma,\Gamma') &= N^sN'^{1-s}C_s(\tau_1,\tau_1')
\le N^sN'^{1-s} C_s(p,p') = C_s(\Lambda,\Lambda'),
\\
D(\Gamma\|\Gamma') &=  N' - N + N \ln \frac{N}{N'} + N D(\tau_1\|\tau_1')
\ge N' - N + N \ln \frac{N}{N'} + N D(p\|p') = D(\Lambda\|\Lambda'),
\end{align}
which also lead to the bounds via
Propositions~\ref{prop_uhlmann}--\ref{prop_KLclassical}.
\end{proof}

Last but not the least, I present propositions concerning the Helstrom
information and the Fisher information \cite{hayashi,helstrom}, which
play crucial roles in parameter estimation.  Let $\epsilon$ and
$\tau_1$ be functions of an unknown vectoral parameter
$\theta \equiv (\theta_1,\dots,\theta_q) \in \Theta \subseteq \mathbb
R^q$. It follows that $\rho_M$ and $\Gamma$ are also functions of
$\theta$.
\begin{proposition}
\label{prop_helstrom}
Define the $q\times q$ Helstrom information matrix as
\begin{align}
K_{\mu\nu}(\rho_M) &\equiv \trace \bk{\sigma_\mu\circ\sigma_\nu} \rho_M,
\label{K}
\end{align}
where $a\circ b \equiv (ab+ba)/2$ denotes the Jordan product and
$\sigma_\mu$ is a symmetric logarthmic derivative (SLD) of $\rho_M$,
defined as a Hermitian-operator solution to
\begin{align}
\parti{\rho_M}{\theta_\mu} &= \sigma_\mu \circ \rho_M.
\label{SLD}
\end{align}
For a Poisson state, $K$ is given by
\begin{align}
K_{\mu\nu}(\rho) &= K_{\mu\nu}(\Gamma) = \trace \bk{S_\mu \circ S_\nu}\Gamma,
\label{K_Gamma}
\end{align}
where $S_\mu$ is an SLD of $\Gamma$, viz., a Hermitian-operator
solution to
\begin{align}
\parti{\Gamma}{\theta_\mu} &= S_\mu \circ \Gamma.
\label{SLD_Gamma}
\end{align}
\end{proposition}
\begin{proof}
  Given the tensor product in Eqs.~(\ref{rare}), it is known that
  \cite{hayashi}
\begin{align}
K(\rho_M) = M K(\tau).
\end{align}
Given the direct sum in Eqs.~(\ref{rare}), it can also be shown that
\begin{align}
K_{\mu\nu}(\tau)&= \sum_l \trace \bk{S_{\mu}^{(l)}\circ S_{\nu}^{(l)}}\pi_l \tau_l,
\end{align}
where $S_\mu^{(l)}$ is an SLD of $\pi_l\tau_l$.
%\begin{align}
%\parti{(\pi_l\tau_l)}{\theta_\mu} &= S_{\mu}^{(l)} \circ \bk{\pi_l\tau_l}.
%\end{align}
Then
$S_\mu^{(0)} = (\partial \pi_0/\partial\theta_\mu)/\pi_0 =
-(\partial\epsilon/\partial\theta_\mu)/(1-\epsilon)$,
$S_\mu^{(1)} = S_{\mu}$, and
\begin{align}
K_{\mu\nu}(\rho_M) &= \frac{1}{M(1-\epsilon)}\parti{N}{\theta_\mu} \parti{N}{\theta_\nu}
+ \trace \bk{S_\mu \circ S_\nu}\Gamma.
\end{align}
Taking the Poisson limit leads to the proposition.
\end{proof}
\begin{proposition}[{well known; see, for example,
    Refs.~\cite{falk11,snyder_miller}}]
\label{prop_fisher}
The Fisher information matrix
\begin{align}
J_{\mu\nu}(P_{\mathcal M}) &\equiv \sum_{m} P_{\mathcal M}
\parti{\ln P_{\mathcal M}}{\theta_\mu}\parti{\ln P_{\mathcal M}}{\theta_\nu}
\end{align}
for a Poisson distribution is given by
\begin{align}
J(P_{\mathcal M}) &= J(\Lambda).
\end{align}
\end{proposition}
Again, with the Poisson limit, the Helstrom and Fisher information
quantities observe a self-similar feature---they are given by the same
formulas as the general ones, except that the unnormalized $\Gamma$ or
$\Lambda$ is substituted into each.

The information quantities obey a monotonicity relation, which follows
from the monotonicity of Helstrom information in two ways.
\begin{proposition}
\label{prop_mono_fisher}
\begin{align}
K(\rho) = K(\Gamma) &\ge J(\Lambda) = J(P_{\mathcal M}),
\end{align}
in the sense that $K(\rho)-J(P_{\mathcal M})$ and
$K(\Gamma)-J(\Lambda)$ are positive-semidefinite.
\end{proposition}
\begin{proof}
  The monotonicity \cite{hayashi} holds on two levels:
  $K(\rho) \ge J(P_{\mathcal M})$, and also
  $K(\Gamma) \ge J(\Lambda)$, which can be proved independently by
  expressing $\Gamma$ and $\Lambda$ in terms of $N$, $\tau_1$, and
  $p = \Lambda/N$ and applying the monotonicity relation with respect
  to $\tau_1$ and $p$.
\end{proof}

\section{\label{sec_fock}Poisson states versus Fock states}

It is important to emphasize that the Poisson theory is in general
different from the usual theory for the Fock state
$\tau_1^{\otimes L}$. For example, from Eq.~(\ref{uhlmann_poisson}),
the fidelity for the Poisson states can be expressed in terms of
$\tau_1$ and $\tau_1'$ as
\begin{align}
F(\rho,\rho') &= \exp\Bk{-\frac{N+N'}{2} + \sqrt{NN'}F(\tau_1,\tau_1')},
\end{align}
which is quite different from
\begin{align}
F(\tau_1^{\otimes L},\tau_1'^{\otimes L}) &= \Bk{F(\tau_1,\tau_1')}^L.
\end{align}
As $\Gamma$ and $\Gamma'$ are not normalized, their fidelity is
bounded as
\begin{align}
0 &\le F(\Gamma,\Gamma') =\sqrt{NN'}F(\tau_1,\tau_1') \le \sqrt{NN'}.
\end{align}
When $F(\tau_1,\tau_1') = 1$ and $F(\Gamma,\Gamma') = \sqrt{NN'}$,
$F(\rho,\rho')$ in the Poisson theory is still less than 1 if
$N \neq N'$, because the different expected object numbers still
lead to distinguishability. On the other hand, if
$F(\tau_1,\tau_1') = 0$ and $F(\Gamma,\Gamma') = 0$, $F(\rho,\rho')$
is still positive, because both states contain the identical $\tau_0$.
Another example is the Helstrom information
\begin{align}
K_{\mu\nu}(\rho) &= 
K_{\mu\nu}(\Gamma) = N\parti{\ln N}{\theta_\mu}\parti{\ln N}{\theta_\nu}
+ N K_{\mu\nu}(\tau_1),
\label{K_tau1}
\end{align}
which is different from
\begin{align}
K(\tau_1^{\otimes L}) = L K (\tau_1).
\end{align}
It is not difficult to prove that, on the per-object basis,
\begin{align}
\frac{K(\Gamma)}{N} &\ge K(\tau_1),
\end{align}
as $N$ may also depend on $\theta$ and the total object number may
give extra information.

In the context of optics, the thermal state in an ultraviolet limit
can be shown to approach the Poisson state, with the mutual coherence
matrix in statistical optics \cite{mandel} becoming a matrix
representation of $\Gamma$ \cite{stellar,tnl,tsang19a}. Alternative
analyses of the exact thermal state have yielded results that are
consistent with the Poisson theory
\cite{tsang19,nair_tsang16,lupo,lu18}.  In the context of partially
coherent imaging, Propositions~\ref{prop_helstrom} and
\ref{prop_fisher} are consistent with our treatment in
Ref.~\cite{tsang_comment19}, although many other studies on this topic
\cite{larson18,larson19,hradil19,liang21,hradil21} compute the
Helstrom information using $K(\tau_1)$ only and may underestimate the
amount of information.  The Poisson model is more realistic than the
Fock model because the former can account for the effects of
inefficiency and loss, which are unavoidable for sources imaged with a
finite aperture \cite{kurdzialek21}.  The use of $K(\tau_1)$ is
justified only if $N$ does not depend on $\theta$, in which case
$K(\Gamma) = NK(\tau_1)$.  I illustrate these concepts with a concrete
example.

\begin{example}
\label{exa_imaging}
Consider the imaging of two equally bright partially coherent optical
sources via a diffraction-limited system in one dimension
\cite{tsang19a,larson18,larson19,hradil19,liang21,hradil21,tsang_comment19,wadood21,de21,kurdzialek21}.
The intensity operator can be modeled as \cite{tsang_comment19}
\begin{align}
\Gamma &= N_0 
\bk{\ket{\psi_1}\bra{\psi_1} + \ket{\psi_2}\bra{\psi_2}
+\gamma \ket{\psi_1}\bra{\psi_2} + \gamma^*\ket{\psi_2}\bra{\psi_1}},
\label{Gamma_imaging}
\\
\ket{\psi_1} &= \intall dx \psi\bk{x + \frac{\theta}{2}}\ket{x},
\quad
\ket{\psi_2} = \intall dx \psi\bk{x - \frac{\theta}{2}}\ket{x},
\label{psi12}
\end{align}
where $N_0 \in \mathbb R_{\ge 0} \equiv \{x \in \mathbb R|x \ge 0\}$
is the expected photon number from one isolated source,
$\gamma \in \mathbb C$ is the degree of coherence with
$|\gamma| \le 1$, $\psi:\mathbb R \to \mathbb C$ is the point-spread
function of the imaging system, $\ket{x}$ is the Dirac position
eigenket that obeys $\braket{x|x'} = \delta(x-x')$, and
$\theta \in \mathbb R_{\ge 0}$ is the separation between the two
sources in Airy units. The expected total photon number is
\begin{align}
N &= \trace\Gamma = 2N_0 \Bk{1 + \real\bk{\gamma\braket{\psi_2|\psi_1}}}.
\end{align}
When $\gamma = 0$ (incoherent sources), $N = 2N_0$ is independent of
$\theta$, but when $\gamma \neq 0$, $N$ may depend on $\theta$, as the
waves from the two sources can interfere, thus enhancing or
suppressing the radiation energy.  For example, when $\gamma = 1$
(fully coherent and in-phase sources) and $\theta = 0$, $N = 4N_0$,
which is consistent with the elementary fact that $n$ identical
in-phase sources with zero separations should radiate at a power
$\propto n^2$ \cite{eberly72}.

To compute the Helstrom information $K(\Gamma)$, I assume the Gaussian
point-spread function
\begin{align}
\psi(x) &= \frac{1}{(2\pi)^{1/4}}\exp\bk{-\frac{x^2}{4}}
\label{gauss_psi}
\end{align}
and use the steps detailed in Appendix~\ref{app_imaging} to compute
Eqs.~(\ref{K_Gamma}) and (\ref{SLD_Gamma}) numerically.  The results
are plotted in Fig.~\ref{imaging_helstrom}.  Almost identical results
are recently reported in Ref.~\cite[Fig.~2(a)]{kurdzialek21}, which
uses the rare-photon model given by Eqs.~(\ref{rare}) and a theory
consistent with the one here.  The use of the intensity operator here
is arguably more direct and convenient, however, as
Eq.~(\ref{Gamma_imaging}) comes naturally from optics and can be used
directly in Eqs.~(\ref{K_Gamma}) and (\ref{SLD_Gamma}), without
returning to the rare-photon model or treating $\epsilon$ and $\tau_1$
separately in the calculations.

\begin{figure}[htbp!]
\centerline{\includegraphics[width=0.6\textwidth]{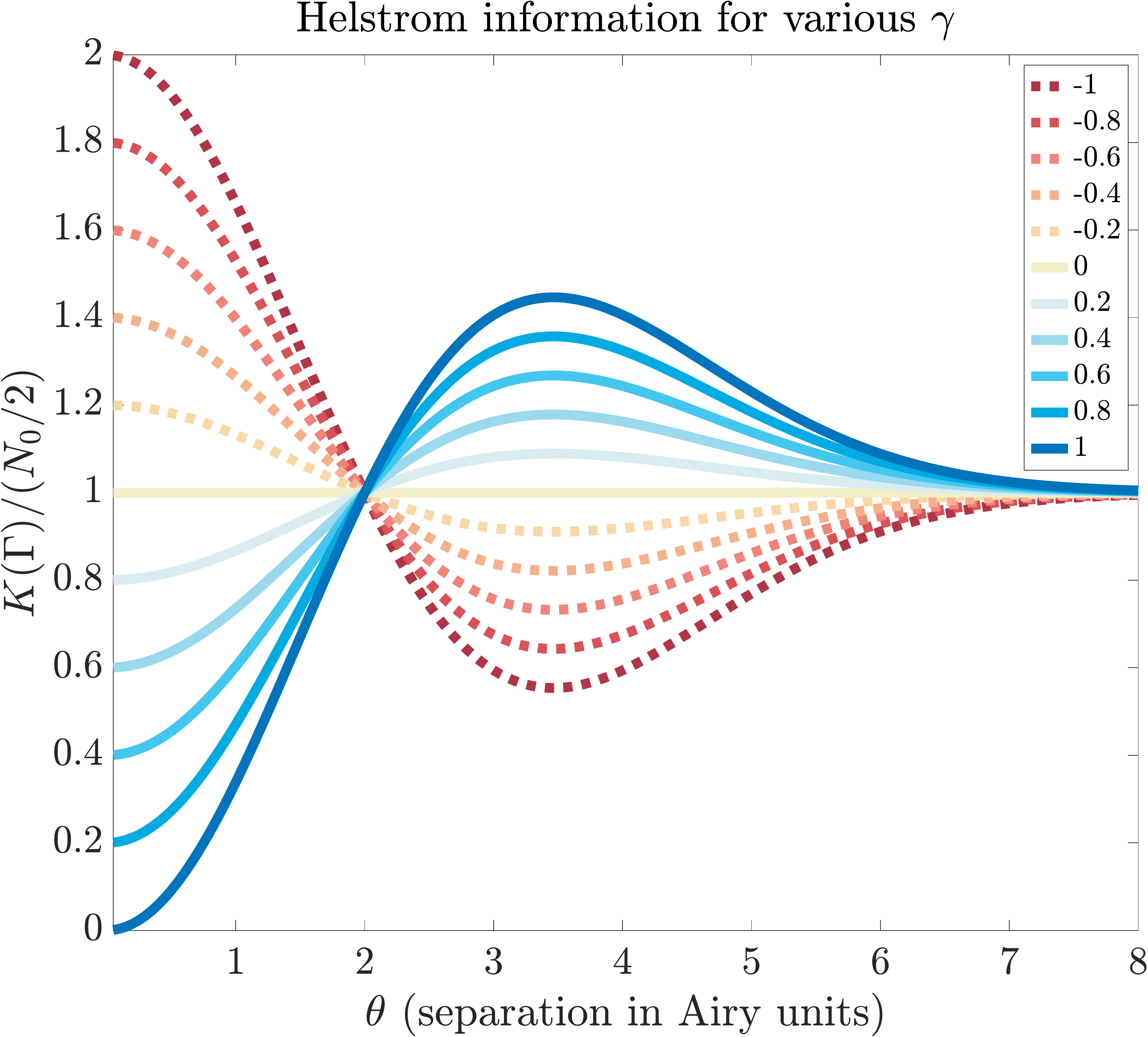}}
\caption{\label{imaging_helstrom}Numerically computed Helstrom
  information $K(\Gamma)$ for the estimation of the separation
  $\theta$ between two partially coherent sources. Each curve assumes
  a fixed degree of coherence $\gamma$, as denoted by the legend.}
\end{figure}

As also noticed in Ref.~\cite{kurdzialek21}, the Helstrom information
plotted in Fig.~\ref{imaging_helstrom} appears to be identical to the
Fisher information with the Hermite-Gaussian measurements plotted in
Ref.~\cite[Fig.~1]{tsang_comment19} for all values of $\gamma$,
suggesting that the measurements are optimal for any $\gamma$ and not
just for the $\gamma = 0$ case proven in Ref.~\cite{tnl}.
\end{example}

\section{\label{sec_channel}Poisson channels}
For the Poisson theory to remain useful for problems involving quantum
channels, the channels should preserve the Poissonianity of a state.
I call such channels Poisson channels. A trace-preserving completely
positive (TPCP) map $\Phi$ on $\rho_M$ is Poisson if
\begin{align}
\Phi\bk{\rho_M} &= \Bk{(1-\epsilon'')\tau_0'' \oplus \epsilon'' \tau_1''}^{\otimes M},
\label{poisson_map}
\end{align}
where $\tau_l'' \in \mathcal P_1(\mathcal H_l'')$, $\mathcal H_0''$ is
1-dimensional, and $\epsilon'' = O(\epsilon)$, such that the output is
still a Poisson state in the Poisson limit. Let $\rho(\Gamma)$ denote
the Poisson state with intensity operator $\Gamma$ in the sense of
Eqs.~(\ref{rare})--(\ref{intensity_op}).  The input-output relation in
the Poisson limit can be abbreviated as
\begin{align}
\Phi\Bk{\rho(\Gamma)} &= \rho\Bk{\tilde\Phi(\Gamma)},
&
\tilde\Phi(\Gamma) &= N''\tau_1'',
&
N'' &\equiv M\epsilon'',
\label{poisson_map2}
\end{align}
where
$\tilde\Phi : \mathcal P(\mathcal H_1) \to \mathcal P(\mathcal H_1'')$
is a map from an intensity operator to another intensity operator
induced by $\Phi$. Then the information quantities
$d_B(\Gamma,\Gamma')$, $D_s(\Gamma,\Gamma')$, $D(\Gamma\|\Gamma')$,
and $K(\Gamma)$ in Eqs.~(\ref{bures}), (\ref{alpha_div_Gamma}),
(\ref{KL_Gamma}), and (\ref{K_Gamma}) also observe monotonicity
relations with respect to $\tilde\Phi$.

\begin{proposition}
\label{prop_mono_gen}
Let $\tilde\Phi$ be a positive map defined in the sense of
Eqs.~(\ref{poisson_map}) and (\ref{poisson_map2}) with respect to a
TPCP map $\Phi$ in the Poisson limit. Then
\begin{align}
d_B(\Gamma,\Gamma') &\ge
d_B\Bk{\tilde\Phi(\Gamma),\tilde\Phi(\Gamma')},
\\
D_s(\Gamma,\Gamma') &\ge 
D_s\Bk{\tilde\Phi(\Gamma),\tilde\Phi(\Gamma')},
\\
D(\Gamma\|\Gamma')
&\ge 
D\Bk{\tilde\Phi(\Gamma)\|\tilde\Phi(\Gamma')},
\\
K(\Gamma) &\ge K\Bk{\tilde\Phi(\Gamma)}.
\quad
(\textrm{if $\tilde\Phi$ does not depend on $\theta$})
\end{align}
\end{proposition}
\begin{proof}
  These monotonicity relations follow from the monotonicity of the
  information quantities with respect to the density operators in
  Eqs.~(\ref{bures}), (\ref{alpha_div_Gamma}), (\ref{KL_Gamma}), and
  (\ref{K_Gamma}).
\end{proof}

To construct examples of $\tilde\Phi$ in the following, I assume
further that $\Phi$ is local, in the sense of
\begin{align}
\Phi(\rho_M) &= \Bk{\phi\bk{\tau}}^{\otimes M},
\label{map_form}
\end{align}
where
$\phi:\mathcal P(\mathcal H_0\oplus \mathcal H_1) \to \mathcal
P(\mathcal H_0''\oplus \mathcal H_1'')$ is a TPCP map that gives
\begin{align}
\phi(\tau) &= \bk{1-\epsilon''}\tau_0'' \oplus \epsilon'' \tau_1''.
\label{poisson_preserve}
\end{align}
Under the local assumption, $\tilde\Phi$ is affine, as shown by
Proposition~\ref{prop_affine} below, although Example~\ref{exa_loss}
later demonstrates that $\tilde\Phi$ need not be trace-preserving and
Example~\ref{exa_se} demonstrates that $\tilde\Phi$ need not be
linear.
\begin{proposition}
\label{prop_affine}
Given Eqs.~(\ref{map_form}) and (\ref{poisson_preserve}), $\tilde\Phi$
is affine.
\end{proposition}
\begin{proof}
  Let the Kraus form \cite{hayashi,watrous} of $\phi$ be
\begin{align}
\phi(\tau) &= \sum_{\alpha} A_\alpha \tau A_\alpha^\dagger.
\end{align}
Expressing $\tau$ and $A_\alpha$ in the matrix forms
\cite[Eq.~(1.64)]{watrous}
\begin{align}
\tau &= \begin{pmatrix}(1-\epsilon)\tau_0 & 0 \\ 0 & \epsilon \tau_1\end{pmatrix},
&
A_\alpha &= \begin{pmatrix}A_{\alpha 00} & A_{\alpha 01}\\ A_{\alpha 10} & A_{\alpha 11}
\end{pmatrix},
\end{align}
where $0$ denotes a zero operator and $A_{\alpha lm}$ is an operator
that maps $\mathcal H_m$ to $\mathcal H_l''$, it can be shown that
\begin{align}
\epsilon'' \tau_1'' &= \sum_\alpha \Bk{A_{\alpha 10}(1-\epsilon)\tau_0 A_{\alpha 10}^\dagger
+ A_{\alpha 11} \epsilon \tau_1 A_{\alpha 11}^\dagger}.
\label{map}
\end{align}
For the $\epsilon'' = O(\epsilon)$ requirement to be satisfied,
it is necessary for the map to satisfy
\begin{align}
\trace \sum_\alpha A_{\alpha 10}\tau_0 A_{\alpha 10}^\dagger &= O(\epsilon),
\label{weak_map}
\\
\trace\sum_\alpha A_{\alpha 11}\tau_1 A_{\alpha 11}^\dagger &= O(1).
\end{align}
Multiplying both sides of Eq.~(\ref{map}) by $M$ and taking the
Poisson limit, the output intensity operator becomes
\begin{align}
\tilde\Phi(\Gamma) &= 
\Gamma'+ \sum_\alpha A_{\alpha 11}\Gamma A_{\alpha 11}^\dagger,
\label{linear_map}
\end{align}
where $\Gamma'$ is the Poisson limit of
$M(1-\epsilon)\sum_\alpha A_{\alpha 10}\tau_0 A_{\alpha 10}^\dagger
\to M\sum_\alpha A_{\alpha 10}\tau_0 A_{\alpha 10}^\dagger$,
which does not depend on $\Gamma$. Equation~(\ref{linear_map}) is
affine with respect to $\Gamma$ (though not necessarily linear).
\end{proof}
In the following examples, the Kraus form of each $\phi$ is most
conveniently expressed in terms of some orthonormal bases of
$\mathcal H_0$, $\mathcal H_0''$, $\mathcal H_1$, and
$\mathcal H_1''$, which are written as $\{\ket{0}\}$, $\{\ket{0''}\}$,
$\{\ket{1_j}|j = 1,\dots,d\}$, and
$\{\ket{1_k''}|k = 1,\dots, d''\}$, respectively. The $\tau$ state,
for example, can be expressed as
\begin{align}
\bk{1-\epsilon}\tau_0 \oplus \epsilon \tau_1 &= \bk{1-\epsilon}\ket{0}\bra{0}
+ \epsilon \sum_{j,k} g_{jk}\ket{1_j}\bra{1_k},
\label{bases}
\end{align}
where $g$ is the density matrix of $\tau_1$ with respect to the basis
above.

\begin{example}[object-number-preserving channel]
  Let a Kraus form of $\phi$ be
\begin{align}
\phi(\tau) &= A_0 \tau A_0^{\dagger} + \sum_{\alpha} A_\alpha \tau A_\alpha^{\dagger},
\\
A_0 &= \ket{0''}\bra{0},
\quad
A_\alpha = \sum_{j,k} A_{\alpha jk}\ket{1_j''}\bra{1_k},
\end{align}
such that 
\begin{align}
\phi\Bk{\bk{1-\epsilon}\tau_0 \oplus \epsilon\tau_1} 
&= \bk{1-\epsilon}\tau_0'' \oplus \epsilon\phi_1(\tau_1),
\end{align}
where 
\begin{align}
\phi_1(\tau_1) &= \sum_{\alpha} A_\alpha \tau_1 A_\alpha^\dagger
\end{align}
is a TPCP map. In the Poisson limit,
\begin{align}
\tilde\Phi(\Gamma) &= \phi_1(\Gamma).
\label{preserve}
\end{align}
\end{example}
The trace-preserving nature of $\phi_1$ means that the channel
preserves the object number. More specific examples include the
unitary map
\begin{align}
\phi_1(\Gamma) &= U \Gamma U^\dagger,
\label{unitary}
\end{align}
which models a lossless linear device via a unitary operator $U$, and
the map
\begin{align}
\phi_1(\Gamma) &= \Lambda,
&
\Lambda_j &= \trace E_j \Gamma,
\label{POVM}
\end{align}
which models a 100\%-efficient object-counting measurement via the
POVM $E$ that gives Eqs.~(\ref{poisson}) and
(\ref{intensity}). Propositions~\ref{prop_mono} and
\ref{prop_mono_fisher} are hence special cases of
Proposition~\ref{prop_mono_gen}.

As the object number need not be conserved, $\tilde\Phi$ need not be
trace-preserving. The next examples of $\tilde\Phi$ are noteworthy
departures from the usual TPCP maps in quantum information theory.
\begin{example}[loss]
\label{exa_loss}
Let $\mathcal H_l'' = \mathcal H_l$ and
\begin{align}
\phi(\tau) &= I_0\tau I_0 + \sum_{j=1}^d A_j \tau A_j^\dagger + T \tau T,
\\
I_0 &= \ket{0}\bra{0},
\quad
A_j = \sqrt{1-\eta_j}\ket{0}\bra{1_j},
\quad
T = \sum_{j=1}^d \sqrt{\eta_j}\ket{1_j}\bra{1_j},
\end{align}
where $0\le \eta_j \le 1$ is the transmission coefficient for each mode.
Then
\begin{align}
\phi\Bk{(1-\epsilon)\tau_0\oplus\epsilon \tau_1} 
&= \bk{1-\epsilon \trace T \tau_1 T}\tau_0
\oplus \epsilon T \tau_1 T.
\end{align}
In the Poisson limit,
\begin{align}
\tilde\Phi(\Gamma) &= T \Gamma T.
\label{loss}
\end{align}
This $\tilde\Phi$ is not trace-preserving if any $\eta_j < 1$,
although it is still completely positive \cite[Theorem~2.22]{watrous}.
\end{example}

\begin{example}[spontaneous emission]
\label{exa_se}
Let $\mathcal H_l'' = \mathcal H_l$ and
\begin{align}
\phi(\tau) &= A_0 \tau A_0^\dagger + \sum_{j=1}^d A_j \tau A_j^\dagger + 
I_1 \tau I_1,
\\
A_0 &= \sqrt{1-\epsilon'}\ket{0}\bra{0},
\quad
A_j = \sqrt{\epsilon'\tau_{1 jj}'}\ket{1_j}\bra{0},
\\
I_1 &= \sum_{j=1}^d \ket{1_j}\bra{1_j},
\label{I1}
\end{align}
where $\epsilon' = O(\epsilon)$ and
$\tau_1' = \sum_j \tau_{1 jj}' \ket{1_j}\bra{1_j} \in \mathcal
P_1(\mathcal H_1)$ are properties of the channel and $\{\ket{1_j}\}$
is assumed to be the eigenbasis of $\tau_1'$ without loss of
generality.  Then
\begin{align}
\phi\Bk{\bk{1-\epsilon}\tau_0\oplus \epsilon\tau_1} &= 
\bk{1-\epsilon}\bk{1-\epsilon'}\tau_0 \oplus 
\Bk{\epsilon\tau_1 + \epsilon'(1-\epsilon)\tau_1'}.
\end{align}
In the Poisson limit,
\begin{align}
\tilde\Phi(\Gamma) &= \Gamma + \Gamma',
\label{se}
\end{align}
where $\Gamma' = M\epsilon'\tau_1'$. In particular, if the input
intensity operator is the zero operator $0$,
\begin{align}
\tilde\Phi(0) &= \Gamma'.
\label{se0}
\end{align}
If $\Gamma' \neq 0$, Eq.~(\ref{se0}) implies that this $\tilde\Phi$
is not trace-preserving and not even linear, as a trace-preserving or
linear map on the zero operator must give the zero operator.
\end{example}
Setting $\epsilon'$, the spontaneous-emission probability per temporal
mode, to be $O(\epsilon)$ ensures that the occurrence of objects
remains rare and $N'' = N + N'$ remains finite under the Poisson
limit.  In reality, however, a channel with a high
spontaneous-emission probability may make the final state non-Poisson
and may also involve other processes such as stimulated emission, so
one should double-check the accuracy of the Poisson approximation
before applying the Poisson theory to a real spontaneous-emission
channel.

Equation~(\ref{se}) is a quantum analog of the fact that the sum of
two independent Poisson processes is also a Poisson process.  It may
be useful for modeling dark counts or background noise
\cite{len20,lupo20,oh21}.  Another version of this fact is the
following.
\begin{example}[composition]
The tensor product of two Poisson states may be expressed as
\begin{align}
\rho(\Gamma) \otimes \rho(\Gamma') &= \rho(\Gamma \oplus \Gamma').
\label{join}
\end{align}
To derive this relation, consider
\begin{align}
\tau \otimes \tau' &= \Bk{(1-\epsilon)(1-\epsilon') \tau_0 \otimes \tau_0'}
\oplus \Bk{\epsilon(1-\epsilon') \tau_1 \otimes \tau_0'}
\oplus\Bk{\epsilon'(1-\epsilon) \tau_0 \otimes \tau_1'}
\oplus \bk{\epsilon\epsilon' \tau_1 \otimes \tau_1'}.
\label{tensor_rare}
\end{align}
Ignore all $O(\epsilon^2)$ terms in Eq.~(\ref{tensor_rare}), including
the two-object component $\epsilon\epsilon' \tau_1 \otimes
\tau_1'$. The output intensity operator in the Poisson limit becomes
\begin{align}
\bk{\Gamma \otimes \tau_0'}\oplus \bk{\tau_0\otimes \Gamma'}.
\label{join_Gamma}
\end{align}
Since $\tau_0=\ket{0}\bra{0}$ and $\tau_0' = \ket{0'}\bra{0'}$ are
1-dimensional, any tensor product of an operator with $\tau_0$ or
$\tau_0'$ is isomorphic to the original operator, and
Eq.~(\ref{join_Gamma}) can be abbreviated as $\Gamma \oplus \Gamma'$.
\end{example}

\begin{example}[marginalization]
  Let $\Gamma$ be an intensity operator on
  $\mathcal H_1 \oplus \mathcal H_1'$ and $\tau$ be a density operator
  on
  $\bk{\mathcal H_0 \oplus \mathcal H_1} \otimes \bk{\mathcal H_0'
    \oplus \mathcal H_1'}$. Let
\begin{align}
\phi(\tau) &= \trace'\tau
\end{align}
denote the partial trace with respect to
$\mathcal H_0' \oplus \mathcal H_1'$. A Kraus form is
\begin{align}
\phi(\tau) &= \bra{0'}\tau\ket{0'} + \sum_{j=1}^{d'}\bra{1_j'}\tau\ket{1_j'}.
\end{align}
Expressing $\tau$ in terms of the basis
$\{\ket{0}\otimes\ket{0'}, \ket{1_j}\otimes\ket{0'},
\ket{0}\otimes\ket{1_k'}| j = 1,\dots,d, k = 1,\dots,d'\}$, it is not
difficult to show that
\begin{align}
\tilde\Phi(\Gamma) &= I_1 \Gamma I_1,
\label{margin}
\end{align}
where $I_1$ is the projection operator into $\mathcal H_1$ given by
Eq.~(\ref{I1}).
\end{example}
Equations~(\ref{preserve}), (\ref{loss}), (\ref{se}), (\ref{join}),
and (\ref{margin}) in the examples demonstrate that, for Poisson
channels, the maps can be much simplified if expressed in terms of the
intensity operators. Whether there exists a more comprehensive
mathematical treatment of Poisson channels in the spirit of quantum
channel theory is an interesting open problem.

\section{Conclusion}
In conclusion, I have shown that the Poisson limit leads to elegant
results in quantum information theory, with the intensity operator
emerging as the central quantity. The familiar appearances of the
formulas mean that one may borrow existing results from general
information theory and apply them to the intensity operators in the
study of Poisson states. Although the unnormalized nature of the
operators may require extra care, one can still take advantage of many
known results concerning unnormalized positive-semidefinite matrices
\cite{dhillon07,amari16,cichocki10,watrous}, as the formulas here
coincide with many of them.

There are many potential generalizations and extensions.  It should be
possible to define the Poisson limit more generally for a tensor
product of non-identical states and an infinite-dimensional $\Gamma$,
in analogy with the more general Poisson limit theorems
\cite{falk11,snyder_miller}. The mathematical rigor of the limit may
be improved through quantum stochastic calculus \cite{parthasarathy92}
or nonstandard analysis \cite{leitz01,nelson87}. The mathematical
theory of Poisson channels may be refined further.

In terms of applications, some specialized aspects of the Poisson
theory have already found success in the study of weak thermal light
for optical sensing and imaging \cite{helstrom,tsang19a}, but a
consolidation of the results under the umbrella of Poisson quantum
information may reveal new insights.  Given the importance of the
classical Poisson theory in diverse areas \cite{falk11,snyder_miller},
the quantum Poisson theory is envisioned to see wider applications in
quantum technologies beyond optics, wherever a sequence of rare
objects may be encountered.

\section*{Acknowledgments}
I acknowledge helpful discussions with Zden\v{e}k Hradil, Nick
Vamivakas, Vincent Tan, and Marco Tomamichel.  This work is supported
by the National Research Foundation (NRF) Singapore, under its Quantum
Engineering Programme (Award~QEP-P7).

\appendix
\section{\label{app_poisson}Poisson limit theorem}
To state the Poisson limit theorem rigorously, I first rephrase the
measurement model in terms of measure theory.  Consider a multi-output
object-counting measurement of the $k$th temporal mode.  Let
$\mathcal Y$ denote the set of outputs and $\mathcal B$ denote its
$\sigma$-algebra. For example, for an ideal direct-imaging
measurement, each $(x,y) \in \mathcal Y = \mathbb R^2$ is a possible
position of the object on the image plane, while each
$B \in \mathcal B$ is a region of the image plane. Let
$X_k \in \{0,1\}$ be the total object number detected by the
measurement. Assume
\begin{align}
\expect\Bk{X_k = 0} &= \trace I_0\tau = 1-\epsilon,
\end{align}
where $\expect$ denotes the expectation, 
$[\textrm{statement}]$ is the Iverson bracket defined as
\begin{align}
\Bk{\textrm{statement}} &\equiv \begin{dcases}
1, & \textrm{statement is true},
\\
0, & \textrm{otherwise},
\end{dcases}
\end{align}
and $\expect[\textrm{statement}]$ is the probability that the
statement is true. Conditioned on $X_k = 1$, let $Y_k \in \mathcal Y$
be the output that detects the object. Assume
\begin{align}
\expect\Bk{Y_k \in B} &= \trace E(B)\tau_1,
\end{align}
where $E:\mathcal B \to \mathcal P(\mathcal H_1)$ is a POVM on
$\mathcal H_1$. The object count in a set of outputs can then be
modeled by the random measure
$\mathcal N^{(k)}:\mathcal B \to \{0,1\}$, defined as
\begin{align}
\mathcal N^{(k)}(B) &\equiv X_k\Bk{Y_k \in B}.
\end{align}
$X_k$ and $Y_k$ can be taken as independent random variables.
Integrated over $M$ temporal modes, the random measure becomes
\begin{align}
\mathcal M_M(B) &\equiv \sum_{k=1}^M \mathcal N^{(k)}(B)
= \sum_{k=1}^M X_k\Bk{Y_k \in B} = 
\sum_{j=1}^{L_M} \Bk{Z_j \in B},
\end{align}
where
\begin{align}
L_M &\equiv \sum_{k=1}^M X_k
\end{align}
is the detected object number in total and the set
$\{Z_1,\dots,Z_{L_M}\}$ is simply a relabeling of $\{Y_k|X_k = 1\}$,
which is a set of $L_M$ independent and identically distributed
(i.i.d.)  random elements. As $X$ is Bernoulli, $L_M$ is binomial. In
the Poisson limit, $L_M$ becomes Poisson, denoted by $L$, and
$\mathcal M_M \to \mathcal M$ is called the Kac empirical point
process \cite{vaart96}.
\begin{theorem}[see, for example, {Ref.~\cite[Sec.~3.5.2]{vaart96}}]
\label{thm_poisson}
With a Poisson $L$ and a sequence of i.i.d.\ random elements
$\{Z_1,Z_2,\dots\}$ that are independent of $L$, the Kac process
$\mathcal M:\mathcal B \to \mathbb N_0$, defined as
\begin{align}
\mathcal M(B) \equiv \sum_{j=1}^{L} \Bk{Z_j \in B},
\label{kac}
\end{align}
is a Poisson process.
\end{theorem}
The intensity measure $\Lambda:\mathcal B \to \mathbb R_{\ge 0}$ of
the Poisson process is given by
\begin{align}
\Lambda(B) &= \expect\Bk{\mathcal M(B)} = N \trace E(B)\tau_1 = \trace E(B)\Gamma.
\end{align}
In the main text, the abbreviations
$\mathcal N_j^{(k)} = \mathcal N^{(k)}(B_j)$,
$\mathcal M_j = \mathcal M(B_j)$, $E_j = E(B_j)$, and
$\Lambda_j = \Lambda(B_j)$ are used, where $\{B_1,B_2,\dots\}$ is a
disjoint partition of $\mathcal Y$. $\{\mathcal M_j\}$ are independent
Poisson random variables by a basic property of Poisson processes, so
their probability distribution is given by Eq.~(\ref{poisson}).

\section{\label{app_state}Poisson states on Fock spaces}
$(\mathcal H_0\oplus\mathcal H_1)^{\otimes M}$ is called a toy Fock
space in Refs.~\cite{parthasarathy92,meyer95}.  The Poisson limit of
$\rho_M$ on the toy Fock space should remain well defined if one is
willing to adopt nonstandard analysis, taking $M$ to be an unlimited
natural number and $\epsilon = N/M$ to be infinitesimal
\cite{leitz01,nelson87}. In the context of ordinary calculus, one may
think of $\epsilon$ as $O(dt)$, where $dt$ is an infinitesimal
interval in time or a similar continuous degree of freedom.  In the
Poisson channel theory, Eq.~(\ref{weak_map}) implies that each Kraus
operator $A_{\alpha 10}$ should scale as
$O(\sqrt{\epsilon}) = O(\sqrt{dt})$, which is why quantum stochastic
calculus may be needed to define the operators rigorously in the
Poisson limit.

A more standard way to deal with an infinite number of temporal modes
in quantum stochastic calculus is to abandon the toy Fock space and
use instead the Fock space $\mathcal F(\mathcal L^2[0,1]\otimes\mathcal H_1)$,
where $\mathcal F$ is defined as \cite{parthasarathy92}
\begin{align}
\mathcal F(\mathcal H) &\equiv 
\bigoplus_{l=0}^\infty \mathcal H^{\otimes l},
&
\mathcal H^{\otimes 0} &= \mathcal H_0,
\end{align}
and the $\mathcal L^2[0,1]$ space, modeling the temporal modes for one object,
is defined as
\begin{align}
\mathcal L^2[0,1] &\equiv \BK{f\middle| f:[0,1] \to \mathbb C,\Avg{f,f} < \infty},
\\
\Avg{f,g} &\equiv \int_0^1 f^*(t)g(t) dt.
\end{align}
See Refs.~\cite{parthasarathy92,meyer95,leitz01} for discussions of
the relation between the toy Fock space and
$\mathcal F(\mathcal L^2[0,1]\otimes\mathcal H_1)$. It is an open problem,
outside the scope of this paper, how Poisson states may be expressed
on $\mathcal F(\mathcal L^2[0,1]\otimes\mathcal H_1)$.

If the objects are collected in one place and their arrival times are
ignored, one can assume a simpler Fock space
\begin{align}
\mathcal F(\mathcal H_1) &= \bigoplus_{l=0}^\infty \mathcal H_1^{\otimes l}
\end{align}
and transform $\rho_M$ into a state on $\mathcal F(\mathcal H_1)$ given by
\begin{align}
\Pi(\rho_M) &= \bigoplus_{l=0}^M P_{L_M}(l) \tau_1^{\otimes l},
\label{fock}
\\
P_{L_M}(l) &= \begin{pmatrix}M\\ l\end{pmatrix} \pi_0^{M-l}\pi_1^l,
\end{align} 
where $\Pi$ denotes an appropriate map, ${L_M}$ is the binomial random
variable for the total object number, and
$\tau_1^{\otimes 0} = \tau_0$ is assumed.  This representation may be
useful because any function of density operators that satisfies a
monotonicity relation with respect to TPCP maps can be computed using
$\Pi(\rho_M)$ instead of $\rho_M$, by virtue of the following
proposition.
\begin{proposition}
\label{prop_fock}
$\Pi$ is TPCP. Moreover, there exists another TPCP map $\Pi'$ such
that $\Pi'[\Pi(\rho_M)] = \rho_M$.
\end{proposition}
The proof is delegated to Appendix~\ref{app_proofs}.

In the Poisson limit, ${L_M}$ becomes Poisson, viz.,
\begin{align}
P_{L_M}(l) &\to P_{L}(l) = \exp(-N) \frac{N^l}{l!},
\end{align}
and the state on $\mathcal F(\mathcal H_1)$ becomes
\begin{align}
\Pi(\rho_M) &\to \Pi(\rho) = 
\bigoplus_{l=0}^\infty P_{L}(l)\tau_1^{\otimes l}
= \exp(-N) \bigoplus_{l=0}^\infty \frac{\Gamma^{\otimes l}}{l!}.
\label{rep2}
\end{align}
$\Pi(\rho)$ has the spirit of a Kac process and may serve as a
rigorous definition of a Poisson state without resorting to
nonstandard analysis. For example, a Poisson process can arise
directly from a measurement on $\Pi(\rho)$ as follows. Let $L$ be the
object number determined by a measurement in terms of the Kraus
operators $\{I_0,I_1,\dots\}$, where $I_l$ is the projection operator
into $\mathcal H_1^{\otimes l}$.  The conditional state becomes
$\tau_1^{\otimes {L}}$. If each object is then measured by the POVM
$E$, the outcomes are the i.i.d.\ random elements
$\{Z_1,\dots,Z_{L}\}$.  The Kac process
$\mathcal M(B) = \sum_{j=0}^{L}[Z_j \in B]$ is therefore a Poisson
process by Theorem~\ref{thm_poisson}.

Many other results in this paper can also be rederived from
Eq.~(\ref{rep2}).  Equations~(\ref{rare}) are more intuitive from the
physics point of view, however, and are therefore used in the main
text.

\section{\label{app_proofs}Proofs of Propositions~\ref{prop_qchernoff}, \ref{prop_KL}, and
  \ref{prop_fock}}
\begin{proof}[Proof of Proposition~\ref{prop_qchernoff}]
To derive Eq.~(\ref{qchernoff_poisson}), take similar steps
to those in the proof of Proposition~\ref{prop_uhlmann} to obtain
\begin{align}
C_s(\rho_M,\rho_M') &= \Bk{C_s(\tau,\tau')}^M
= \Bk{\sum_l C_s(\pi_l\tau_l,\pi_l'\tau_l')}^M
\\
&= 
\Bk{1 - s\epsilon - (1-s)\epsilon' + \epsilon^s \epsilon'^{1-s}
C_s(\tau_1,\tau_1') + O(\epsilon^2)}^M,
\end{align}
and then take the Poisson limit.

Note that Eq.~(\ref{qchernoff_distance}) takes the infimum of the
Poisson limit, but to prove that it is the same as the Poisson limit
of the infimum, one may need to prove the uniform convergence of
$C_s(\rho_M,\rho_M')$ to $C_s(\rho,\rho')$ over $0\le s\le 1$, beyond
the pointwise convergence just proved. This complication suggests that
the Poisson limit is an intricate mathematical issue, and the
Fock-space representation given by Eq.~(\ref{rep2}), the Poisson limit
of which has already been taken, may be a better starting point for
rigorous proofs. In particular, Eq.~(\ref{rep2}) leads directly to
Eqs.~(\ref{qchernoff_poisson}) and (\ref{qchernoff_distance}) without
any ambiguity.
\end{proof}

\begin{proof}[Proof of Proposition~\ref{prop_KL}]
Consider
\begin{align}
D(\rho_M\|\rho_M') &= M D(\tau\|\tau') 
\label{KLiid}
\\
&= M \sum_l \trace \pi_l\tau_l \Bk{\ln(\pi_l\tau_l) - \ln (\pi_l'\tau_l')}
\label{KLdirectsum}
\\
&= M \trace (1-\epsilon)\tau_0 \BK{\ln\Bk{(1-\epsilon)\tau_0 }
-\ln\Bk{(1-\epsilon')\tau_0 }}
+ M\trace \epsilon \tau_1 \Bk{\ln\bk{\epsilon \tau_1}-
\ln\bk{\epsilon' \tau_1'}}
\\
&= M(1-\epsilon)\ln\frac{1-\epsilon}{1-\epsilon'}
+ \trace \Gamma\bk{\ln \Gamma-\ln\Gamma'},
\label{KLintensity}
\end{align}
where Eq.~(\ref{KLiid}) has used Ref.~\cite[Eq.~(5.97)]{watrous} and
Eq.~(\ref{KLintensity}) has used Ref.~\cite[Eq.~(5.99)]{watrous}.  The
Poisson limit of the first term in Eq.~(\ref{KLintensity}) is
\begin{align}
M(1-\epsilon)\ln\frac{1-\epsilon}{1-\epsilon'}
&= 
(1-\epsilon)\ln\Bk{1+\epsilon'-\epsilon + O(\epsilon^2)}^M\to N' - N,
\end{align}
which leads to Eq.~(\ref{KLpoisson}).
\end{proof}

\begin{proof}[Proof of Proposition~\ref{prop_fock}]
  To prove that $\Pi$ is TPCP, I construct a physical procedure that
  takes $\rho_M$ as input and produces $\Pi(\rho_M)$ as output.
\begin{enumerate}
\item Designate the input system as system A and an auxiliary system
  with Hilbert space $(\mathcal H_0\oplus\mathcal H_1)^{\otimes M}$
  and initial state $\tau_0^{\otimes M}$ as system B.

\item For each $k = 1,\dots,M$:
\begin{enumerate}
\item Measure the number of objects in the $k$th temporal mode of
  system A with the Kraus operators $\{I_0,I_1\}$. Let
  $X_k \in \{0,1\}$ be the random variable from the measurement.
\item If $X_k = 0$, do nothing.
\item If $X_k = 1$, put the object detected in the $k$th temporal mode
  of system A (in conditional state
  $I_1\tau I_1/(\trace I_1\tau I_1)$) into an empty temporal
  mode of system B (any mode in state $\tau_0$).
\end{enumerate}
\item Discard all the empty temporal modes in system B so that its
  Hilbert space becomes $\mathcal H_1^{\otimes L_M}$, where
  $L_M \equiv \sum_{k=1}^M X_k$ is the number of detected objects.

\item Give system B as the output.
\end{enumerate}
With $\rho_M$ as the input state, $X$ is Bernoulli, $L_M$ is binomial,
the output state conditioned on $L_M$ is $\tau_1^{\otimes L_M}$, and the
unconditional output state becomes Eq.~(\ref{fock}).

To prove that a TPCP map giving $\Pi'[\Pi(\rho_M)] = \rho_M$ exists, I
construct a physical procedure that takes $\Pi(\rho_M)$ as input and
gives $\rho_M$ as output.
\begin{enumerate}
\item Assume a system A with Hilbert space
  $(\mathcal H_0\oplus \mathcal H_1)^{\otimes M}$ and initial state
  $\tau_0^{\otimes M}$. Assume a system B with Hilbert space
  $\oplus_{l=0}^M \mathcal H_1^{\otimes l}$.  For the $\Pi'$ map,
  assume that system A is auxiliary and system B is the input.

\item Measure the number of objects $L_M$ in system B with the Kraus
  operators $\{I_0,I_1,\dots\}$.

\item Generate a random classical bit sequence
  $X = \{X_1,\dots,X_M\} \in \{0,1\}^M$ with $L_M$ bits equal to $1$ and
  $M-L_M$ bits equal to $0$. The probability of each sequence is assumed
  to be
\begin{align}
P_{X|L_M}(x|l) &= \begin{pmatrix}M\\ l\end{pmatrix}^{-1} \Bk{\sum_k x_k=l}.
\end{align}
\item For each $k = 1,\dots,M$: 
\begin{enumerate}
\item If $X_k = 0$, do nothing.

\item If $X_k = 1$, put an object in system B into the $k$th temporal
  mode of system A.
\end{enumerate}
\item Give system A as the output.
\end{enumerate}
With $\Pi(\rho_M)$ as the input state, $L_M$ is binomial, and the
unconditional probability of each bit sequence becomes
\begin{align}
P_X(x) &= \sum_{l=0}^M  P_{L_M}(l) P_{X|L_M}(x|l)
= \sum_{l=0}^M \pi_1^l\pi_0^{M-l}\Bk{\sum_k x_k=l}
= \prod_{k=0}^M  \bk{[x_k=0]\pi_0 + [x_k=1]\pi_1},
\end{align}
meaning that $X$ is Bernoulli.  Conditioned on $X$, the output state
is
\begin{align}
\rho_M(X) &= \bigotimes_{k=1}^M \bk{[X_k= 0]\tau_0 \oplus [X_k=1] \tau_1}.
\end{align}
The unconditional output state is hence 
\begin{align}
\expect\Bk{\rho_M(X)}
&= \bigotimes_{k=1}^M \bk{\pi_0\tau_0 \oplus \pi_1\tau_1} = \rho_M.
\end{align}
\end{proof}

\section{\label{app_imaging}Computation of the Helstrom information
for Example~\ref{exa_imaging}}
To compute Eqs.~(\ref{K_Gamma}) and (\ref{SLD_Gamma}) numerically, I
follow the method in Refs.~\cite{genoni19,peng21}. Express each
operator in the problem as
\begin{align}
A &= \sum_{j,k}\tilde A_{jk}\ket{\psi_j}\bra{\psi_k},
\end{align}
where the set $\{\ket{\psi_j}\}$ spans $\mathcal H_1$. The set need
not be orthogonal or normalized. Let
$\Delta \equiv \partial\Gamma/\partial\theta$ for a scalar parameter
$\theta$.  Then Eq.~(\ref{SLD_Gamma}) can be expressed as
\begin{align}
2\tilde{\Delta} &= \tilde S G \tilde\Gamma + \tilde\Gamma G \tilde S,
\label{Stilde}
\end{align}
where $G_{jk} = \braket{\psi_j|\psi_k}$, while Eq.~(\ref{K_Gamma}) can
be expressed as
\begin{align}
K &= \trace G\tilde S G \tilde\Delta.
\label{Ktilde}
\end{align}
For Example~\ref{exa_imaging}, $\ket{\psi_1}$ and $\ket{\psi_2}$ are
defined in Eqs.~(\ref{psi12}). Define also
\begin{align}
\ket{\psi_3} &\equiv \parti{\ket{\psi_1}}{\theta},
&
\ket{\psi_4} &\equiv \parti{\ket{\psi_2}}{\theta}.
\end{align}
Then
\begin{align}
\tilde\Gamma &= N_0\begin{pmatrix}1 & \gamma & 0 & 0\\
\gamma^* & 1 & 0 & 0\\
0 & 0 & 0 & 0\\
0 & 0 & 0 & 0\end{pmatrix},
&
\tilde\Delta &= N_0\begin{pmatrix} 0 & 0 & 1 &\gamma\\
0 & 0 & \gamma^* & 1\\
1 & \gamma & 0 & 0\\
\gamma^* & 1 &0  & 0
\end{pmatrix}.
\end{align}
With the Gaussian $\psi(x)$ given by Eq.~(\ref{gauss_psi}), $G$ can
also be computed analytically (with the help of the Symbolic Math
Toolbox on Matlab (ver.~2020b, Mathworks)):
\begin{align}
G &= \begin{pmatrix} 1 & \exp(-\frac{\theta^2}{8}) & 0 & 
-\frac{\theta}{8}\exp(-\frac{\theta^2}{8})\\
\exp(-\frac{\theta^2}{8}) & 1  &-\frac{\theta}{8}\exp(-\frac{\theta^2}{8}) & 0\\
0 & -\frac{\theta}{8}\exp(-\frac{\theta^2}{8}) & \frac{1}{16} & 
\frac{\theta^2-4}{64}\exp(-\frac{\theta^2}{8})\\
-\frac{\theta}{8}\exp(-\frac{\theta^2}{8})& 0 & 
\frac{\theta^2-4}{64}\exp(-\frac{\theta^2}{8})& \frac{1}{16}\end{pmatrix}.
\end{align}
For each $\gamma$ from $-1$ to $1$ with step size $0.2$ and each
$\theta$ from $0.05$ to $8$ with step size $0.05$, Eq.~(\ref{Stilde})
is solved for $\tilde S$ using the \texttt{lyap} function on Matlab
and $K$ is computed using Eq.~(\ref{Ktilde}). The results are
plotted in Fig.~\ref{imaging_helstrom}.

To prevent the error ``\texttt{The solution of this Lyapunov equation
  does not exist or is not unique}'' in using the \texttt{lyap}
function on Matlab, a tiny positive number $\delta$ is artificially
introduced to $\tilde\Gamma_{33}$ and $\tilde\Gamma_{44}$.  To check
that $\delta$ does not affect the results significantly, the numerical
analysis is repeated with various $\delta$, and it is found that the
plots do not show any perceptible change with $\delta$ ranging from
$10^{-13}$ to $10^{-5}$. Figure~\ref{imaging_helstrom} uses
$\delta = 10^{-13}$.

The Matlab code is reproduced below.  The \texttt{lbmap} routine to
generate a color map for colorblind viewers can be downloaded from
Ref.~\cite{bemis21}.

\begin{verbatim}
function helstrom
theta = .05:.05:8;
gamma = -1:.2:1;
delta = 1e-13;
hold off;
for m = 1:length(gamma)
  Gamma = [1 gamma(m) 0 0; conj(gamma(m)) 1 0 0; 0 0 delta 0; 0 0 0 delta];
  dGamma = [0 0 1 gamma(m); 0 0 conj(gamma(m)) 1; ...
            1 gamma(m) 0 0; conj(gamma(m)) 1 0 0];
  for n = 1:length(theta)
    G11 = 1;
    G12 = exp(-theta(n)^2/8);
    G13 = 0;
    G14 = -theta(n)/8*exp(-theta(n)^2/8);
    G22 = 1;
    G23 = G14;
    G24 = 0;
    G33 = 1/16;
    G34 = (theta(n)^2-4)/64*exp(-theta(n)^2/8);
    G44 = 1/16;
    G = [G11 G12 G13 G14; conj(G12) G22 G23 G24; ...
         conj(G13) conj(G23) G33 G34; conj(G14) conj(G24) conj(G34) G44];
    S = lyap(Gamma*G,-2*dGamma);
    K(n) = 2*trace(G*S*G*dGamma);
  end
  if gamma(m) < 0
    plot(theta,K,':');
  else
    plot(theta,K);
  end
  hold on;
end
axis([min(theta) max(theta) 0 2]);
legend(num2str(gamma.'));
xlabel('$\theta$ (separation in Airy units)');
ylabel('$K(\Gamma)/(N_0/2)$') ;
title('Helstrom information for various $\gamma$');
newcolor = lbmap(length(gamma),'RedBlue');
colororder(newcolor);
\end{verbatim}

\bibliographystyle{apsrev4-1}
\bibliography{research2}

\end{document}